\documentclass[pdftex,12pt,preprint,sort&compress]{elsarticle}

\usepackage{makeidx,hyperref}
\usepackage{graphicx}
\usepackage{amssymb,amsmath,amsthm}
\usepackage{natbib}
\newcommand{\D}{{\mathrm{d}}}
\newtheorem{theorem}{Theorem}
\newtheorem{lemma}{Lemma}

\newtheorem{corollary}{Corollary}
\newtheorem{remark}{Remark}
\usepackage{epigraph}

\begin{document}

\begin{frontmatter}

\title{New universal Lyapunov functions\\ for non-linear reaction networks}
\author{Alexander N. Gorban}
\ead{a.n.gorban@le.ac.uk}
\address{Department of Mathematics, University of Leicester, Leicester, LE1 7RH, UK \\ and Lobachevsky University, Nizhny Novgorod, Russia}

\begin{abstract}
In 1961, R\'enyi discovered a rich family of non-classical Lyapunov functions for kinetics of the Markov chains,
or, what is the same, for the linear kinetic equations. This family was parameterized by convex functions on the positive semi-axis. After works of Csisz\'ar and Morimoto,
these functions became widely known as $f$-divergences or the Csisz\'ar--Morimoto divergences.
These Lyapunov functions  are universal in the following sense: they depend only on the state of equilibrium,
not on the kinetic parameters themselves.

Despite many years of research, no such wide family of universal Lyapunov functions has been found for nonlinear reaction networks. For general non-linear networks with detailed or complex balance, the classical thermodynamics potentials remain the only universal Lyapunov functions.

We constructed a rich family of new universal Lyapunov functions for {\em any
non-linear reaction network} with detailed or complex balance.
These functions are parameterized by compact subsets of the projective space.
They are universal in the same sense: they depend only on the state of equilibrium and on the network structure, but not on the kinetic parameters themselves.

The main elements and operations in the construction of the new Lyapunov functions
are partial equilibria of reactions and convex envelopes of families of functions.

\end{abstract}


\date{}

\end{frontmatter}

\section{Introduction}

The classical Lyapunov functions in kinetics are closely related to the concepts of entropy and free energy. The first example was provided by Boltzmann in 1872 \cite{Boltzmann1872}. He  proved that in an isolated system the functional 
\begin{equation}
H(f)=\int f(x,v)\ln(f(x,v)) \D^3 v \D^3x
\end{equation}
does not increase in time due to kinetic equation (the Boltzmann equation). Here, $f(x,v)$ is the distribution density of particles in the phase space, $x$ is position in space, $v$ is velocity of a particle. Boltzmann's proof \cite{Boltzmann1872} used the principle of detailed balance. Several years later he found more general conditions for $H$-theorem  \cite{Boltzmann1887} and invented what we call now semidetailed balance or cyclic balance or complex balance (for modern theory of chemical reaction networks  we refer to tutorial \cite{Angeli2009}).

For perfect chemical mixtures with components $A_1, \ldots, A_n$ in isothermal isochoric conditions (fixed volume) the analogue of Boltzmann's $H$-function is:
\begin{equation}\label{LyapFreeEN}
H=\sum_{i=1}^n c_i \left(\ln\left(\frac{c_i}{c_i^{\rm eq}}\right)-1 \right),
\end{equation}
where $c_i\geq 0$ is the concentration of $A_i$ and $c_i^{\rm eq}>0$ is an equilibrium concentration of $A_i$ (under the standard convention that $x\ln x=0$ for $x=0$). The idealization of perfect mixtures is applicable to rarefied gases or to reactions of small admixtures  in solutions.

The $H$-function  (\ref{LyapFreeEN})  is a Lyapunov function for {\em all} mass action law systems with detailed or complex balanced equilibrium (see e.g. \cite{Angeli2009})). We call this property `universality'.

In information theory, the function $H$ appears as a measure of relative information (in the distribution $c_i$ with respect to the distribution $c_i^{\rm eq}$) and analogue of the $H$-theorem states that random manipulations with data decrease the relative information with respect to the equilibrium that does not change under manipulations \cite{Shannon1948, Cohen1993, CohenIwasa1993}.

It is not much surprising that the $H$ function (\ref{LyapFreeEN}) is essentially the only universal Lyapunov function for all imaginable perfect kinetic systems with detailed balance. Nevertheless, if we restrict the choice of  the reaction mechanism then the class of Lyapunov functions, which are conditionally independent of reaction rate constants for a given detailed balanced or complex balanced equilibrium, can be extended. We   call such Lyapunov functions conditionally universal (for a given reaction mechanism).

In 1961,  R\'enyi discovered a class of conditionally universal Lyapunov functions for Markov chains \cite{Renyi1961}. After works \cite{Csiszar1963, Morimoto1963} these functions were studied by many authors under the name {\em $f$-divergences} or Csisz\'ar--Morimoto divergences. It is known that any  universal Lyapunov functions for Markov chains has the form of $f$-divergence \cite{ENTR3, Amari2009,Gorban2Judge2010} or is a monotonic function of such a divergence. The continuous time Markov kinetic equation coincide with the kinetic equations for linear (monomolecular) reactions of perfect systems, and $f$-divergences can be considered as a direct generalization of (\ref{LyapFreeEN}):
\begin{equation}\label{F-div}
H_f(c)=\sum_i c_i^{\rm eq}f\left(\frac{c_i}{c_i^{\rm eq}}\right),
\end{equation}
where $f$ is a convex function on the positive semi-axis.

For $f(x)=x\ln x$ and after adding a constant  term proportional to $\sum_i c_i$ we get the classical formula (\ref{LyapFreeEN}). (Recall that $\sum_i c_i$ does not change in linear kinetics.)

Existence of a very rich family of conditionally universal Lyapunov functions for the linear reaction mechanisms makes us guess that there should be many conditionally universal Lyapunov functions for any given nonlinear reaction mechanism as well. In this paper, we construct new conditionally universal Lyapunov functions for any given reaction mechanism using partial equilibria of all single reactions and detailed or complex balance conditions.

In the next Sec. \ref{Prerequisit}, we give the necessary formal definitions and introduce notations for mass action law systems. The necessary and sufficient conditions that a convex function is a Lyapunov function for all reaction networks with given reaction mechanism and equilibrium point under detailed or complex balance assumption are proven in Sec.\ref{Sec:GenTheorem}. The construction of  a new family of conditionally universal Lyapunov functions for any reaction network  is  presented and the main result, Theorem \ref{HTheoremNew}, is proven for mass action law systems in Sec. \ref{secNewMAL}. In Sec. \ref{General:Sec} we  outline the possible generalizations and applications of the results.  In Conclusion the main results of the work  are summarized and an open question is formulated.

\section{Prerequisites: mass action law and classical Lyapunov functions \label{Prerequisit}}

In this section, we formally introduce mass action law and equations of chemical kinetics. For more detailed introduction, including thermodynamical backgrounds, detailed kinetics, applied kinetics, and mathematical aspect of kinetics, we refer to the modern book \cite{MarinYablonsky2019}. Tutorial \cite{Angeli2009} gives the mathematical introduction in dynamics of chemical reaction networks. Formalism of chemical kinetics with special attention to heterogeneous catalysis is discussed in detail in the monograph \cite{Yab}.

\subsection{Mass action law}

Consider a closed system with $n$ chemical species $A_1,\ldots, A_n$, participating in a complex reaction network. The  reaction network is represented   in the form of the system of
{\em stoichiometric equations} of elementary reactions (called also {\em reaction mechanism}):
\begin{equation}\label{ReactMech}
\sum_{i=1}^n \alpha_{ri} A_i \to \sum_{j=1}^n \beta_{rj} A_j \;\; (r=1, \ldots, m) \, ,
\end{equation}
where  $\alpha_{ri}\geq 0$, $\beta_{rj}\geq 0$ are the stoichiometric coefficients, $r=1,\ldots, m$, $i,j=1, \ldots, n$, $m$ is the number of elementary reactions, $n$ is the number of components. In this representation, the direct and reverse elementary reactions are considered separately.

The {\em stoichiometric vector} $\gamma_r$ of the elementary
reaction is $\gamma_r=(\gamma_{ri})$, $\gamma_{ri}=\beta_{ri}-\alpha_{ri}$ (`gain minus loss'). The gain vector is $\alpha_r=(\alpha_{ri})$ and the loss vector is $\beta_r=\beta_{ri}$.

Elementary reactions of the form $A_i\to A_j$ are called linear or monomolecular reactions.

According to the {\em  mass action law}, the reaction rate for the elementary reactions (\ref{ReactMech}) are
\begin{equation}\label{GenMAL}
w_r=k_r \prod_{i=1}^n c_i^{\alpha_{ri}} ,
\end{equation}
where $k_r\geq 0$ is the {\em reaction rate constant} and the standard convention is used: for any $x\geq 0$, $x^0=1$.

The kinetic equations for a perfect system in isochoric isothermal conditions have the form
\begin{equation}\label{kinurChem}
\frac{\D c}{\D t}=\sum_{r=1}^m \gamma_r w_r
\end{equation}

The {\em stoichiometric subspace} is $Span\{\gamma_r |r=1, \ldots, m\}$. For any set of values  $k_r\geq 0$ and $c_i \geq 0$ the time derivative of $c$ belongs to the stoichiometric subspace
$$\frac{\D c}{\D t} \in Span\{\gamma_r |r=1, \ldots, m\}.$$
Therefore, each linear functional $b(c)$ that annuls  $Span\{\gamma_r |r=1, \ldots, m\}$ (i.e. $b(\gamma_r)=0$ for all $r$ is the conservation law:
 $\D b(c)/\D t=0$ according to (\ref{kinurChem}). Such functionals are called {\em stoichiometric conservation law}. 

We always assume that there exists a strictly positive stoichiometric conservation law that is such  a vector $b=(b_i)$, $b_i>0$ that $\sum_i b_i \gamma_{ri}=0$ for all $r$. In chemical kinetics, this may be the conservation of mass or of total number of atoms, for example. Due to this assumption, every stoichiometric  vector $\gamma$ has both positive and negative components. This assumption has many important consequences. For example, if 
\begin{equation}\label{alpha>1} 
\mbox{For all } i,r \mbox{ either }\alpha_{ri}=0 \mbox{ or }\alpha_{ri}\geq 1,
\end{equation}
 then the right hand side of (\ref{kinurChem}) is a Lipschitz function in any {\em reaction polyhedron}
\begin{equation}\label{reactpolyhedr}
(Span\{\gamma_r |r=1, \ldots, m\}+c) \cap \mathbb{R}^n_{\geq 0},
\end{equation}
where $c$ is an arbitrary vector with non-negative coordinates and $\mathbb{R}^n_{\geq 0}$ is the cone (orthant) of such vectors. Therefore, solution  of (\ref{kinurChem}) exists  and is unique for every initial state with non-negative coordinates. Moreover, it remains non-negative and, therefore, any reaction polyhedron (\ref{reactpolyhedr}) is positively invariant: solutions do not leave it in positive time. Algorithms for detailed analysis of the structure of reaction polyhedra are developed in \cite{Gorban1984}. It involves the classical double-description problem \cite{Chernikova1965,FukudaProdon2996}. Reaction polyhedra for various examples of chemical reactions are presented in \cite{GorbanKagan2006,GorbanSIADS2013,Pisarenko2018}. The boundedness of positive semi-trajectories has strong stability implications for reaction networks \cite{HangosSzeder2013}.

In many applications, the stoichiometric coefficients are non-negative integers. For global existence, uniqueness of solutions, and preservation of positivity the condition (\ref{alpha>1}) and existence of positive stoichiometric conservation law are sufficient. Moreover, in this work we do not need uniqueness and global existence of solutions and analyze  just the local conditions $\D G/ \D t \leq 0$  for various functions $G$ at strictly positive points $c$. Therefore, a weaker condition is needed: instead of (\ref{alpha>1}) we assume below that $\alpha_{ri}\geq 0$, $\beta_{ri} \geq 0$.	

\subsection{Detailed balance and classical Lyapunov function}

First proof of Boltzmann's $H$-theorem \cite{Boltzmann1872} and its analogues for chemical kinetics were based on detailed balance .  Boltzmann's argument were analyzed by Tolman \cite{Tolman1938}.  In general situation, this principle means that at equilibrium, each elementary process is equilibrated with its reverse process.

 To formulate this principle for chemical kinetics, we combine direct and reverse reactions together and rewrite the reaction mechanism in the form
\begin{equation}\label{ReactMechRev}
\sum_{i=1}^n \alpha_{ri} A_i \rightleftharpoons \sum_{j=1}^n \beta_{rj} A_j \;\; (r=1, \ldots, m) \, ,
\end{equation}
(formally, the transformation from (\ref{ReactMech}) to (\ref{ReactMechRev}) is always possible if we allow some of $k_r^-=0$ when there is no reverse reaction in the initial mechanism).

We use notation $k_r^+$ and $k_r^-$ for reaction rate constants of direct and reverse reactions, and  $w_r^+$, $w_r^-$ for reaction rates of these reactions:
\begin{equation}\label{GenMALDB}
w_r^+=k_r^+ \prod_{i=1}^n c_i^{\alpha_{ri}}, \; \;  w_r^-=k_r^- \prod_{i=1}^n c_i^{\beta_{ri}} ,
\end{equation}
For the total rate of the couple of direct  and reverse reaction we use 	$w_r=w_r^+-w_r^-$. In these notations the kinetic equations are the same (\ref{kinurChem}) (with different value of $m$).

The principle of detailed balance for the mass action law is:  there exists a positive equilibrium $c_i^{\rm eq}>0$ with
detailed balance, 
\begin{equation}\label{DetBal}
w_r^+(c_i^{\rm eq})=w_r^-(c_i^{\rm eq}) \mbox{ for all }r.
\end{equation} 
If the set of the stoichiometric vectors $\{\gamma_r\}$ is linearly dependent then this condition implies algebraic relations between reaction rate constants. 
Each elementary reaction is equilibrated at the point of detailed balance, $c^{\rm eq}$. For systems that obey the mass action law, this means that
\begin{equation}\label{delBalconst}
k_r^+ \prod_{i=1}^n (c_i^{\rm eq})^{\alpha_{ri}}=k_r ^-\prod_{i=1}^n (c_i^{\rm eq})^{\beta_{ri}}=w_r^{\rm eq}.
\end{equation}
It is convenient to use the detailed balance relations (\ref{delBalconst}) and introduce a set of independent parameters instead of the reaction rate constants: the equilibrium fluxes $w_r^{\rm eq}$ and the equilibrium concentrations $c_i^{\rm eq}$. The reaction rate constants have a simple and explicit expression through the equilibrium flows and equilibrium concentrations:
\begin{equation}\label{ChangeK}
k_r^+=w_r^{\rm eq}\left(\prod_{i=1}^n (c_i^{\rm eq})^{\alpha_{ri}}\right)^{-1}, \;
k_r ^-=w_r^{\rm eq}\left(\prod_{i=1}^n (c_i^{\rm eq})^{\beta_{ri}}\right)^{-1}.
\end{equation}
In this parameterization, the mass action law (\ref{GenMALDB}) with detailed balance condition takes the form:
\begin{equation}\label{GenMALMod}
w_r^+=w_r^{\rm eq} \prod_{i=1}^n \left(\frac{c_i}{c_i^{\rm eq}}\right)^{\alpha_{ri}} \,  w_r^-=w_r^{\rm eq} \prod_{i=1}^n \left(\frac{c_i}{c_i^{\rm eq}}\right)^{\beta_{ri}} ,
\end{equation}
In this parameterization, the kinetic equations for the mass action law with detailed balance are
\begin{equation}\label{reparKINUR}
\frac{\D c}{\D t}=\sum_{r=1}^m \gamma_r w_r^{\rm eq} \left[ \prod_{i=1}^n \left(\frac{c_i}{c_i^{\rm eq}}\right)^{\alpha_{ri}}- \prod_{i=1}^n \left(\frac{c_i}{c_i^{\rm eq}}\right)^{\beta_{ri}}\right].
\end{equation}

The classical $H$-theorem for isochoric isothermal perfect chemical systems with detailed balance can be produced now by simple straightforward calculations. Notice that 
\begin{equation}\label{Hgradient}
\frac{\partial H}{\partial c_i}=\ln\left(\frac{c_i}{c_i^{\rm eq}}\right)
\end{equation}
 and calculate the time derivative of $H$ according to the system of kinetic equations with detailed balance condition (\ref{reparKINUR}):
\begin{equation}\label{h-theorem}
\frac{\D H}{\D t}=-\sum_{r=1}^m  (w_r^+-w_r^-)(\ln w_r^+- \ln w_r^-)\leq 0,
\end{equation}
and ${\D H}/{\D t}=0$ for points of detailed balance an only for such points. Therefore, existence  of a positive point of detailed balance for the mass action law kinetic equation  implies that $H$ (\ref{LyapFreeEN}) does not increase in time and all the non-negative steady states of kinetic equations are detailed balance equilibria.

Definition of $H$ and  $c_i^{\rm eq}$ requires additional comments. Equilibrium concentrations of chemical mixture depend on the conserved quantities those do not change in the course of chemical reactions. For example, they depend on the atomic balances. Without fixing these values, the vector of positive equilibria $c_i^{\rm eq}$ ($i=1,\ldots,n$) is defined ambiguously. Any vector of positive equilibria $c_i^{\rm eq}>0$ can be used in the definition of $H$ (\ref{LyapFreeEN}). After that, the $H$-function (\ref{LyapFreeEN}) can be used for all values of $c_i\geq 0$ for all possible values of conserved quantities. It is a simple exercise to show that for the mass action law kinetics with detailed balance the difference between $H$-functions with different choices of equilibria $c_i^{\rm eq}$ does not change in time.

The $H$-function  (\ref{LyapFreeEN}) was utilised by Zeldovich in his proof of uniqueness of positive equilibrium for given values of conserved quantities (1938, reprinted in 1996 \cite{Zeld}). It was recognised as the main instrument for analysis of stability of perfect kinetic systems in 1960s-1970s \cite{ShapiroShapley1965, VolpertKhudyaev1985}. 

\subsection{Complex balance and Lyapunov functions}

In 1887, Lorentz stated that the collisions of polyatomic molecules are irreversible and, therefore, Boltzmann's $H$-theorem is not applicable to the polyatomic media \cite{Lorentz1887}. Boltzmann found the solution immediately and invented what we call now semidetailed balance or cyclic balance or complex balance \cite{Boltzmann1887}.

Now, it is proven that the Lorentz objections were wrong and the detailed balance conditions hold for polyatomic molecules \cite{CercignaniLampis1981}. Nevertheless, this discussion was seminal. The complex balance is a popular assumption in chemical kinetics beyond the detailed balance \cite{Horn1972,Feinberg1972}. The comparative analysis of detailed and complex balance assumption in practice of modeling of chemical reaction networks was presented in work \cite{Szederkenyi2011}. It is demonstrated how the generalized mass action law with complex balance appears as a macroscopic limit of the microscopic Markov kinetics \cite{GorbanShahzad2011}. The formal structures of complex balance are also useful for analysis of systems with time delays \cite{Hangos2018}.

Return to the `irreversible' representation of the reaction mechanism (\ref{ReactMech}), where direct and reverse reaction are considered separately. Let $c^{\rm eq}$ be a positive vector of concentrations (it will be a point of complex balance). Similarly to (\ref{ChangeK}) introduce new parameters
\begin{equation}
k_r=w_r^{\rm eq}\left(\prod_{i=1}^n (c_i^{\rm eq})^{\alpha_{ri}}\right)^{-1}
\end{equation}
$w_r^{\rm eq}$ is the rate of the $r$th reaction at the point $c^{\rm eq}$.
The kinetic equations with these parameters are 
\begin{equation}\label{KinUrComplBal}
\frac{\D c}{\D t}=\sum_{r=1}^m \gamma_r w_r^{\rm eq}\prod_{i=1}^n \left(\frac{c_i}{c_i^{\rm eq}}\right)^{\alpha_{ri}}.
\end{equation}

Let us calculate the time derivative of the function $H$ (\ref{LyapFreeEN}) by virtue of the system (\ref{KinUrComplBal}). 

\begin{equation}\label{entropyproductionGNEKIN}
\frac{\D H}{\D t}=\sum_i \frac{\partial H}{\partial c_i} \frac{\D c_i}{\D t} 
= \sum_{r }(\gamma_{r }, \nabla H)w_r^{\rm eq}\exp(\alpha_r, \nabla H),
\end{equation}
where $\nabla H$ is given by (\ref{Hgradient}) 

An auxiliary function $\theta(\lambda) $ of one variable $\lambda\in [0,1]$ is convenient for analysis of $\D H/ \D t$ (see \cite{Gorban1984,GorbanShahzad2011}):
\begin{equation}\label{auxtheta}
\theta(\lambda)=\sum_{\rho}w_r^{\rm eq}\exp[(\nabla H,(\lambda
\alpha_{r}+(1-\lambda)\beta_{r}))]
\end{equation}
With this function, ${\D H}/{\D t}$ defined by (\ref{entropyproductionGNEKIN}) has a
very simple form:
\begin{equation}\label{EntropProdtheta}
\frac{\D H}{\D t}=-\left.\frac{\D \theta(\lambda)}{\D
\lambda}\right|_{\lambda=1}
\end{equation}
The condition  $\theta(1)=\theta(0)$  is sufficient for the inequality $\theta'(1) \leq 0$, because  $\theta(\lambda)$ is a convex function. Hence, if $\theta(1)=\theta(0)$ then 
${\D H}/{\D t}\leq 0$. The explicit form of the condition   $\theta(1)=\theta(0)$ is
\begin{equation}\label{0=1}
\sum_{r}w_r^{\rm eq}\prod_{i=1}^n \left(\frac{c_i}{c_i^{\rm eq}}\right)^{\alpha_{ri}}=
\sum_{r}w_r^{\rm eq}\prod_{i=1}^n \left(\frac{c_i}{c_i^{\rm eq}}\right)^{\beta_{ri}}.
\end{equation}

Let us consider the family of all gain and loss vectors $\{\alpha_{r},\beta_{r}\}$ ($r=1, \ldots ,m$). Usually, some of these $2m$ vectors coincide. Assume that there are $q$ different vectors among them. Let $y_1, \ldots, y_q$ be these vectors. For each $j=1, \ldots, q$ we take
\begin{equation}
R_j^+=\{r | \,\alpha_{r}=y_j\}\, , \; R_j^-=\{r |
\,\beta_{r}=y_j\}
\end{equation}

We can rewrite  (\ref{0=1}) in the form
\begin{equation}\label{0=1incomplex}
\sum_{j=1}^q \prod_{i=1}^n \left(\frac{c_i}{c_i^{\rm eq}}\right)^{y_{ji}}\left[\sum_{r\in R_j^+}w_r^{\rm eq}- \sum_{r\in R_j^-}w_r^{\rm eq}\right]=0
\end{equation}
The monomials $$\prod_{i=1}^n \left(\frac{c_i}{c_i^{\rm eq}}\right)^{y_{ji}}$$ are linearly independent functions of $c$ for every finite set of vectors $y_j$. Therefore for any $j=1, \ldots, q$
\begin{equation}\label{complexbalanceGENKIN}
\sum_{r\in R_j^+}w_r^{\rm eq}- \sum_{r\in R_j^-}w_r^{\rm eq}=0
\end{equation}
This {\em complex balance condition} is equivalent to the condition $\theta(0)=\theta(1)$ and is sufficient for the inequality ${\D H}/{\D t}\leq 0$. Therefore, $H$ is a Lyapunov function for systems with complex balance.

\section{\label{Sec:GenTheorem}General $H$ theorem for perfect systems with detailed and complex balance}

The systems with detailed or complex balance have a classical Lyapunov function (\ref{LyapFreeEN}) but we are interested in construction of new Lyapunov functions. In this section, we find necessary and sufficient conditions that a convex function is a conditionally universal Lyapunov function for the reaction network with given reaction mechanism and equilibrium point under detailed or complex balance assumption.

Let us start from the systems with detailed balance. By definition, a function $G(c,c^{\rm eq})$ is a conditionally universal Lyapunov function for this reaction network if it is a Lyapunov function for system (\ref{reparKINUR}) for every set of non-negative  values of the equilibrium fluxes  $w_r^{\rm eq} $.

The following lemma is a simple consequence of the form of the kinetic equations (\ref{reparKINUR}). Let $G(c,c^{\rm eq})$ be continuous function and  a convex function of $c$ for all $c^{\rm eq}$.
\begin{lemma}\label{splittingLemma}
Time derivative of a function $G(c,c^{\rm eq})$  by virtue of system (\ref{reparKINUR}) is non-negative, $\D G(c,c^{\rm eq})/\D t\leq 0$, for all positive values of $ w_r^{\rm eq}$, $c$, and  $c^{\rm eq}$ if and only if the derivatives of this function by virtue of the following systems are non-negative for all $r$, all positive   $c$ and given $c^{\rm eq}$:
\begin{equation}\label{1DKINUR}
\frac{\D c}{\D t}= \gamma_r \left[ \prod_{i=1}^n \left(\frac{c_i}{c_i^{\rm eq}}\right)^{\alpha_{ri}}- \prod_{i=1}^n \left(\frac{c_i}{c_i^{\rm eq}}\right)^{\beta_{ri}}\right].
\end{equation}
\end{lemma}
\begin{proof} For smooth $G(c,c^{\rm eq})$, calculate the time derivative of this function by virtue of the  system (\ref{reparKINUR}):
\begin{equation}\label{DG/Dt}
\begin{split}
\frac{\D G(c,c^{\rm eq})}{\D t}&=\sum_{r=1}^m  w_r^{\rm eq} (\nabla_c G(c,c^{\rm eq}), \gamma_r) \left[ \prod_{i=1}^n \left(\frac{c_i}{c_i^{\rm eq}}\right)^{\alpha_{ri}}- \prod_{i=1}^n \left(\frac{c_i}{c_i^{\rm eq}}\right)^{\beta_{ri}}\right] \\
&=\sum_{r=1}^m  w_r^{\rm eq} D_r  G(c,c^{\rm eq}),
\end{split}
\end{equation}
where $ (\nabla_c G(c,c^{\rm eq}), \gamma_r)$ is the standard inner product (derivative of $G(c,c^{\rm eq})$ in  the direction $\gamma_r$) and $ D_r  G(c,c^{\rm eq})$ is the time derivative of  $G(c,c^{\rm eq})$  by virtue of the system (\ref{1DKINUR}).

The coefficients $ w_r^{\rm eq}$ are independent non-negative variables. Therefore, the time derivative of  $G(c,c^{\rm eq})$  by virtue of the  original system (\ref{reparKINUR}) is a conical combination of $ D_r  G(c,c^{\rm eq})$. Non-negativity of all conical combinations of  $ D_r  G(c,c^{\rm eq})$ means that each term is non-negative. Thus, non-negativity of ${\D G(c,c^{\rm eq})}/{\D t}$ for the reaction network with any values of equilibrium fluxes is equivalent to non-negativity of ${\D G(c,c^{\rm eq})}/{\D t}$ for each single reaction subsystem of the network (i.e.  to the inequality  $D_r  G(c,c^{\rm eq})\geq 0$ for all $r$). For continuous convex functions $G(c,c^{\rm eq})$ we have to use in (\ref{SmoothUniversal}) the {\em subgradients}  of $G(c,c^{\rm eq})$ instead of the gradients. (Recall that according to A.D. Alexandrov theorem continuous convex functions are almost everywhere twice differentiable and, therefore, the criterion (\ref{SmoothUniversal}) with classical gradients is valid for them almost everywhere.) The time derivative in (\ref{DG/Dt}) becomes an interval of non-positive numbers and the proof remains the same. For more detail about continuous but non-smooth Lyapunov functions we refer to \cite{Sontag1999,Clarke2001} \end{proof}

Lemma \ref{splittingLemma} allows us to reduce a complex validation of inequalities with $m+2n$ variables $ w_r^{\rm eq}$, $c_i^{\rm eq}$, and $c_i$ to a series of $m$ simpler inequalities with $2n$ variables. Moreover, non-negativity of $D_r  G(c,c^{\rm eq})$ means that the function $G(c,c^{\rm eq})$ does not increase with time along solutions of {\em one-dimensional} kinetic equations (\ref{1DKINUR}): in this system,  ${\D c}/{\D t}$ is proportional to vector $\gamma_r$ and for any positive solution $c(t)$ ($t>0$) the difference $c(t_1)-c(t_0)$ is always proportional to $\gamma_r$, $c(t_1)-c(t_0)=\xi \gamma_r$ with some scalar multiplier $\xi$ ($t_{0,1}>0$). This one-dimensional dynamics admits a strongly convex Lyapunov functions $H$ (\ref{LyapFreeEN}) with time derivative $ -(w_r^+-w_r^-)(\ln w_r^+- \ln w_r^-)$ (\ref{h-theorem}).

For each positive concentration vector $c$, the interval
\begin{equation}\label{intervalI}
I=(c+\mathbb{R}\gamma_r)\cap  \mathbb{R}_{>0}^n
\end{equation}
is positively invariant with respect to (\ref{1DKINUR}).  Restriction of $H$ (\ref{LyapFreeEN}) on this interval is a strongly convex function. The derivative of this function has the logarithmic singularity at the ends of the interval. $H$ has the unique minimizer  on $I$. It is the unique positive equilibrium point of  (\ref{1DKINUR}) on $I$ (and, by definition, a partial equilibrium of the complete system (\ref{reparKINUR})). These observations allow us to formulate the following criterion for the conditionally universal convex Lyapunov functions of the reaction kinetics (\ref{reparKINUR}) for reaction networks with detailed balance.

 Let $G(c,c^{\rm eq})$ be continuous function and  a convex function of $c\in  \mathbb{R}_{>0}^n$ for all $c^{\rm eq}$.
\begin{theorem}\label{GeneralHtheorem}
 $ G(c,c^{\rm eq})$ is a monotonically non-increasing function of time on the positive solutions of the kinetic equations   (\ref{reparKINUR}) for all non-negative values of equilibrium fluxes  if and only if for every positive concentration vector $c$ and every stoichiometric vector  $\gamma_r$ of the reaction mechanism the minimizer of $H(c,c^{\rm eq} )$ on the interval $I$ (\ref{intervalI}) is, at the same time, a minimizer of $ G(c,c^{\rm eq})$ on this interval:
\begin{equation}
\underset{{c+ \gamma_r  x \in \mathbb{R}_{>0}^n}}
{\operatorname{argmin}} H(c+ \gamma_r  x,c^{\rm eq} ) \subset \underset{{c+ \gamma_r  x \in\mathbb{R}_{>0}^n}} {\operatorname{argmin}} G(c+ \gamma_r  x,c^{\rm eq})
\end{equation}
\end{theorem}

\begin{proof}
According to Lemma \ref{splittingLemma}, it is necessary and sufficient to prove this theorem for the one-step reaction ($m=1$). Consider restriction of the one-step kinetic equation (\ref{1DKINUR}) on the interval $I$ (\ref{intervalI}). On this interval, the system has one  equilibrium (the partial equilibrium $c_{\gamma_r}^*$). It is stable, and the restriction of $H$ on this interval is the Lyapunov function of the system. The point $c_{\gamma_r}^*$ is the minimizer of $H$ on $I$.

Assume that  $c_{\gamma_r}^*$ is a minimizer of  $ G(c,c^{\rm eq})$ for $c\in I$. Then convexity of $G(c)$ implies that $G$ is monotonically non-increasing function of time on the solution $c(t)$ of (\ref{1DKINUR}) on $I$. ($G(c(t))$ decreases monotonically to the minimal value when the point $c(t)$ approaches its minimizer $c_{\gamma_r}^*$.)

Assume now that $G(c(t))$ does not increase in time due to dynamic of  (\ref{1DKINUR}) on $I$. This dynamics lead $c(t)$ to the unique equilibrium $c_{\gamma_r}^*$. This equilibrium should be a minimizer of $G$ on $I$. Indeed, if at some point $y\in I$ the function $G$ takes smaller value then in  $c_{\gamma_r}^*$, then in the motion from $y$ to $c_{\gamma_r}^*$ the value of $G$ should increase, which contradicts the assumption.
\end{proof}

Thus, to check that a convex function is a conditionally universal function for the reaction network with detailed balance, it is sufficient to check that its minimizers in the direction of the stoichiometric vectors of the reaction mechanism include the minimizers of $H$ (i.e. one-step partial equilibria).  Of course, this is a much simpler task than analysis of the signs of $\D G/\D t$ for all states and all values of parameters. Nevertheless, this simple check gives necessary and sufficient conditions for a function to be a conditionally universal Lyapunov function for the kinetic equations with a given reaction mechanism.

A cone of possible velocities is a convenient tool for analysis of conditionally universal Lyapunov functions for reaction networks. Consider all systems with detailed balance, a given reaction mechanism and a positive equilibrium $c^{\rm eq}$. According to the mass action law kinetic equations with detailed balance (\ref{reparKINUR}) the time derivatives $\D c / \D t$ at given point $c$ form a cone  $\mathbb{Q}_{\rm DB}(c)$:
\begin{equation}\label{DBcone}
\mathbb{Q}_{\rm DB}(c) = \left. cone\left\{\gamma_r  sign\left[\prod_{i=1}^n \left(\frac{c_i}{c_i^{\rm eq}}\right)^{\alpha_{ri}}-\prod_{i=1}^n \left(\frac{c_i}{c_i^{\rm eq}}\right)^{\beta_{ri}}\right]\; \right| r=1, \ldots , m\right\}, 
\end{equation}
where $cone$ stands for the conic hull and $sign(x)=1$ if $x>0$, $-1$ if $x<0$ and 0 if $x=0$.

\begin{remark}[Differential criterion of conditionally universal Lyapunov functions]
A smooth function $G(c,c^{\rm eq})$ is a conditionally universal Lyapunov functions for a given reaction mechanism and positive equilibrium $c^{\rm eq}$ if
\begin{equation}\label{SmoothUniversal}
(\nabla_c G(c,c^{\rm eq}), x)\leq 0 \mbox{  for all  non-negative  } c \mbox{  and all  } x \in \mathbb{Q}_{\rm DB}(c).
\end{equation}
For continuous  convex functions we have to use in (\ref{SmoothUniversal}) the {\em subgradients} of $G(c,c^{\rm eq})$ instead of gradients (see also the proof of Lemma \ref{splittingLemma}).
\end{remark}

The systems of complex balance are defined by linear relationships  (\ref{complexbalanceGENKIN}) between non-negative values $w_r^{\rm eq}$. Therefore, the time derivatives $\D c / \D t$ at given point $c$  for all systems with complex balanced equilibrium $c^{\rm eq}$ also form a cone, $\mathbb{Q}_{\rm CB}(c)$.

It is obvious that  $\mathbb{Q}_{\rm DB} \subseteq  \mathbb{Q}_{\rm CB}$. Surprisingly, these cones coincide \cite[Theorem 2  (Local equivalence of detailed and complex balance)]{Gorban2014}. This means that for every mass action law system with the complex balanced positive equilibrium and any concentration vector $c$ there exists a mass action law system with detailed balance and the same positive equilibrium such that the velocity vectors $\D c/ \D t$ at point $c$ for these systems coincide. If some reactions of the complex balance system are irreversible, a reverse reaction should be added. Its rate constant will be zero for the original complex balance system and non-zero for the detailed balance system. This theorem was proven in \cite{Gorban2014} even for more general kinetic law, the generalized mass action law. A bit earlier, such theorem was proven for Markov processes \cite{Gorban2013Markov}.  Continuous-time Markov kinetics with a finite number of states and a given positive equilibrium was studied. This class of systems is significantly wider than the systems with detailed balance. Nevertheless, for an arbitrary  probability distribution $P$ and a general system there exists a system with detailed balance and the same equilibrium that has the same velocity $\D P/\D t$ at point $P$. The results are extended to nonlinear systems with the generalized mass action law.

\begin{remark}[Coincidence of conditionally universal Lyapunov functions for detailed and complex balance systems] According to \cite[Theorem 2]{Gorban2014},  a  function $G(c,c^{\rm eq}$ is a Lyapunov function for all systems (\ref{kinurChem}) with given reaction mechanism and complex balanced equilibrium $c^{\rm eq}$ if and only if it is the universal Lapunov functions for all systems (\ref{reparKINUR}) with detailed balance, the same equilibrium and the same mechanism (supplemented by the reverse reactions, if necessary).
\end{remark}

For some reaction mechanisms there exist Lyapunov functions without any relation to detailed balance or complex balance. For example, assume that all the elementary reactions have the form
\begin{equation}\label{noInter}
\alpha_{ri}A_i\to \sum_{j=1}^n \beta_{rj} A_j
\end{equation}
(only one $\alpha_{ri}$ can be non-zero; direct and reverse reactions are considered separately and some reactions can be irreversible). If there exists a positive balance $\sum_i m_i c_i=M=const$ then for any two solutions of the kinetic equations $c^1(t)$, $c^2(t)$ with the same value of $M$ the weighted $l_1$ distance between them $\sum_i m_i |c^1_i(t)-c^2_i(t)|$ monotonically decreases \cite{GorbanBykYab1986}.

Convergent dynamics with quadratic Lyapunov norms $\|x\|^2=(x,Px)$, where $P$ is a symmetric positive definite matrix, was studied by Demidovich in 1960s and widely used \cite{Henk2004}. Systems (\ref{noInter}) give us example of a class of non-linear convergent systems in weighted $l_1$ norm. Some other examples of reaction mechanisms with such convergence property were produced in \cite{Angeli2007,Banaji2013} on the basis of monotonicity idea. All these selected mechanisms are rather simple. They have convergence property for any values of reaction rate constants. The monotonicity idea is useful for selection of reaction kinetic equations with stable dynamics without detailed or complex  balance properties \cite{Angeli2007,Sontag2007}.

On the contrary, in this paper we consider reaction networks with an arbitrary (presumably, nonlinear) reaction mechanism but with specific restrictions on the reaction rate constants. They should obey the principle of detailed balance. (Later on we explain why the same results are valid for systems with the so-called complex balance.) For such systems, there exist thermodynamic Lyapunov functions. For perfect systems under isothermal isochoric conditions the explicit form of this function is presented by (\ref{LyapFreeEN}). The situation with linear kinetics was similar when R\'enyi revealed $f$-divergences for Markov chains. The decrease of relative entropy (information) in time was well-known but there were no other Lyapunov functions until R\'enyi work \cite{Renyi1961}. Below we construct a wide family of additional Lyapunov functions for any nonlinear reaction network, obeying the  mass action law and the principle of detailed balance.

\section{Partial equilibria and new Lyapunov functions for mass action law \label{secNewMAL}}

Let vector   $\gamma$ have   both positive and negative components.
For every vector of concentrations $c$ we define the corresponding {\em partial equilibrium} in direction $\gamma$ as
\begin{equation}\label{partEQ1}
c_{\gamma}^*(c)=\underset{{c+\gamma x \in\mathbb{R}_{>0}^n}}
{\operatorname{argmin}} H(c+ \gamma  x).
\end{equation}
This partial equilibrium $c_{\gamma}^*(c)$ is the minimizer of $H$ on the interval
$$(c+\mathbb{R}\gamma)\cap  \mathbb{R}_{>0}^n.$$
This interval is bounded.  For a positive point $c$ the minimizer  $c_{\gamma}^*(c)$ is also positive. This is an elementary consequence of the logarithmic singularity of $(c\ln c)'$ at zero. Here and below, argmin is the set of points where the function gets its minimum. The functions $H(c)$ is strongly convex on each bounded set because its Hessian has the form
$$\frac{\partial^2 H(c)}{\partial c_i \partial c_j} = \frac{1}{c_i} \delta_{ij},$$
where $ \delta_{ij}$ is the Kronecker delta.
Therefore, each the argmin set in (\ref{partEQ1}) consists of one point.

For monomolecular and for bimolecular reactions there are simple analytic expression for partial equilibria. Consider a monomolecuar reaction $A_i \rightleftharpoons A_j$. The non-zero components of the stoichiometric vector $\gamma$ are: $\gamma_i=-1$, $\gamma_j=1$. For a given vector $c$, the partial equilibrium is given by the equation $k^+c_i^*(c)=k^-c_j^*(c)$ under condition that $c^*(c)=c+x \gamma$. Simple algebra gives:
\begin{equation}\label{linearQE}
\begin{split}
&c_i^*(c)=\frac{k^-}{k^++k^-}\left(c_i+c_j\right);\\
&c_j^*(c)=\frac{k^+}{k^++k^-}\left(c_i+c_j\right).
\end{split}
\end{equation}
Other components of  $c^*(c)$ coincide with those of $c$.
The sum $c_i+c_j=b$ does not change in the reaction $A_i \rightleftharpoons A_j$. Rewrite (\ref{linearQE}) using this `partial balance' $b$:
\begin{equation}\label{linearQEsh}
c_i^*(c)=\frac{k^-b}{k^++k^-}; \;\;c_j^*(c)=\frac{k^+b}{k^++k^-}.
\end{equation}
The `rate constants' in (\ref{linearQE}) and (\ref{linearQEsh})
can be defined through a positive equilibrium point $c^{\rm eq}$: for the linear reaction, $A_i \rightleftharpoons A_j$, $k^+ c_i^{\rm eq}=k^-c_j^{\rm eq}$ and we can take, for example, 
$$k^+=\frac{c_j^{\rm eq}}{c_i^{\rm eq}+c_j^{\rm eq}}; \; k^-=\frac{c_i^{\rm eq}}{c_i^{\rm eq}+c_j^{\rm eq}}$$
(we use the normalization condition $k^++k^-=1$ to select one solution from the continuum of proportional sets of constants). The expression for the partial equilibrium  for the linear reaction (\ref{linearQEsh}) is
\begin{equation}\label{linearQEEQ}
c_i^*(c)=\frac{c_i^{\rm eq}b}{c_i^{\rm eq}+c_j^{\rm eq}}; \;\;c_j^*(c)=\frac{c_j^{\rm eq}b}{c_i^{\rm eq}+c_j^{\rm eq}}.
\end{equation}

For a bimolecular reaction $A_i+A_j \rightleftharpoons A_k$ the non-zero components of the stoichiometric vector $\gamma$ are: $\gamma_i=\gamma_j=-1$, $\gamma_k=1$. Two independent `partial balances'  that do not change in the reaction are:
$$b_1=c_i+c_j+2c_k,\;\; b_2= c_i-c_j.$$

The partial equilibrium $c^*(c)$ is the positive solution of the equation $k^+c_i^*(c)c_j^*(c)=k^-c_k^*(c)$ under condition that $c^*(c)=c+x \gamma$. After solving of quadratic equation for $x$ we get:
\begin{equation}\label{NlinearQE}
\begin{split}
&c_i^*(c)=\frac{b_2}{2}-\frac{k^-}{2k^+}+\sqrt{\frac{b_2^2}{4}+\frac{k^-b_1}{2k^+}+\left(\frac{k^-}{2k^+}\right)^2};\\
&c_j^*(c)=-\frac{b_2}{2}-\frac{k^-}{2k^+}+\sqrt{\frac{b_2^2}{4}+\frac{k^-b_1}{2k^+}+\left(\frac{k^-}{2k^+}\right)^2};\\
&c_k^*(c)=\frac{b_1}{2}+\frac{k^-}{2k^+}-\sqrt{\frac{b_2^2}{4}+\frac{k^-b_1}{2k^+}+\left(\frac{k^-}{2k^+}\right)^2}.
\end{split}
\end{equation}
The signs in front of square root are selected to provide positivity of $c^*(c)$. 

We can rewrite (\ref{NlinearQE}) using any positive equilibrium $c^{\rm eq}$. Indeed, at this point, the detailed balance gives $k^+ c^{\rm eq}_ic^{\rm eq}_j=k^-c^{\rm eq}_k$ and $$\frac{k^-}{k^+}=\frac{c^{\rm eq}_ic^{\rm eq}_j}{c^{\rm eq}_k}.$$
Therefore,
\begin{equation}\label{NlinearQE2}
\begin{split}
&c_i^*(c)=\frac{b_2}{2}-\frac{c^{\rm eq}_ic^{\rm eq}_j}{2c^{\rm eq}_k}+\sqrt{\frac{b_2^2}{4}+b_1\frac{c^{\rm eq}_ic^{\rm eq}_j}{2c^{\rm eq}_k}+\left(\frac{c^{\rm eq}_ic^{\rm eq}_j}{2c^{\rm eq}_k}\right)^2};\\
&c_j^*(c)=-\frac{b_2}{2}-\frac{c^{\rm eq}_ic^{\rm eq}_j}{2c^{\rm eq}_k}+\sqrt{\frac{b_2^2}{4}+b_1\frac{c^{\rm eq}_ic^{\rm eq}_j}{2c^{\rm eq}_k}+\left(\frac{c^{\rm eq}_ic^{\rm eq}_j}{2c^{\rm eq}_k}\right)^2};\\
&c_k^*(c)=\frac{b_1}{2}+\frac{c^{\rm eq}_ic^{\rm eq}_j}{2c^{\rm eq}_k}-\sqrt{\frac{b_2^2}{4}+b_1\frac{c^{\rm eq}_ic^{\rm eq}_j}{2c^{\rm eq}_k}+\left(\frac{c^{\rm eq}_ic^{\rm eq}_j}{2c^{\rm eq}_k}\right)^2}.
\end{split}
\end{equation}
Similar formulas can be easily obtained for the bimolecular reactions $A_i+A_j \rightleftharpoons A_k+A_l$ and $A_i+A_j \rightleftharpoons 2A_k$. They require  nothing more than the quadratic formula and the detailed balance condition.

Partial equilibria appear in non-equilibrium thermodynamics from the very beginning. Already Jaynes considered conditional maximization of entropy as a  basic method of equilibrium and non-equilibrium statistical physics \cite{Jaynes}. According to Grmela, the equilibrium and nonequilibrium thermodynamics as well as the
equilibrium and nonequilibrium statistical mechanics should be considered as 
particular representations of the Dynamical Maximum Entropy Principle \cite{Grmela2013}. Conditional maximization of entropy gives thermodynamic basis to thermodynamic of driven systems \cite{Grmela2016}.  Partial equilibria (or quasiequilibria - conditional maximizers of entropy for more general conditions)  are often considered as constrained equilibria  as a result of introducing external or internal constraints. For example, the models of extended irreversible thermodynamics can be produced by the conditional maximization of entropy subject to various dynamical constrains and hypotheses about slow and fast variables \cite{JouLebon1996}. Methods of invariant manifolds produces dynamical correction to these conditional maximum entropy models \cite{GorbanKarlin2005}.

In this work, we use a very simple case of partial equilibria: these are equilibria with constrains under which only one reversible elementary reaction is on, while other reactions are frozen. To be more precise, we use the term `$H$-function' instead of `entropy'.
 
The {\em partial equilibrium $H$-function} is associated with the partial equilibria:
\begin{equation}\label{partEQEnt1}
H_{\gamma}^*(c)= H(c_{\gamma}^*(c))=\min_{c+\gamma x  \in\mathbb{R}_{>0}^n} H(c+ \gamma  x).
\end{equation}

We do not assume here any time separation and constrained dynamics and use partial equilibria in a completely different way. Dynamics of $c(t)$ remains undeformed and foll;owes the original kinetic equations. For each stoichiometric vector $\gamma$  we consider the partial equilibrium $c^*_{\gamma}(c(t))$ as a projection of the genuine dynamics on the hypersurface of partial equilibria. This set of projections  is a `shadow' of kinetic curves. The new Lyapunov function for the genuine kinetics is combined from the entropies of these projections. For every compact set $\Gamma$  of vectors $\gamma$ with   both positive and negative components, we define
\begin{equation}\label{GreatpartEQEnt}
H_{\Gamma}^*(c)=\max_{\gamma \in \Gamma } H_{\gamma}^*(c).
\end{equation}
This definition has a simple explanation: for a given initial positive concentration vector $c$ and each stoichiometric vector $\gamma \in \Gamma $ we find the partial equilibrium of the one-step system ($m=1$) with this stoichiometric vector. This partial equilibrium, $c_{\gamma}^*(c)$,  is the projection of the initial vector $c$ parallel to the vector $\gamma$ onto the hypersurface of partial equilibria defined by the equation
$$(\nabla H, \gamma)=0.$$ 

From all these projections ($\gamma \in \Gamma$) we select the most non-equilibrium state, i.e., the state with the maximal value of $H=H(c_{\gamma}^*(c))$ (\ref{partEQEnt1}). This maximal value is the new function $H_{\Gamma}^*(c)$ (\ref{GreatpartEQEnt}). It is defined by the set $\Gamma$ that should include all the stoichiometric vectors of the reaction mechanism. 

It is necessary to stress that for every reaction mechanism there exists continuum of compact sets $\Gamma$, which include  the stoichiometric vectors of this mechanism. The partial equilibrium is the same for the vector $\gamma$ and $x \gamma$ for any $x\neq 0$. Therefore, rigorously speaking, the functions $H_{\Gamma}^*(c)$ should be indexed by subsets of the projective space (the space of one-dimensional subspaces), and not by sets of vectors.

The hypersurface of partial equilibria for a given vector $\gamma$ is the hyperplane orthogonal to $\gamma$ in the entropic inner product. The level sets of the partially equilibrium function $H_{\gamma}^*(c)$ are cylinders with the spheric base and axis parallel to $\gamma$. 

Fig. \ref{PatialEquilib} represents a very simplified example of the partial equilibrium $c_{\gamma}^*(c)$ and the Lyapunov function $H_{\Gamma}^*(c)$ construction. The system with three components of the same molecular weight is presented in the triangle $c_1+c_2+ c_3=const$ drawn in barycentric coordinates. The reaction mechanism consists of three reactions $A_1 \rightleftharpoons A_2$,$A_2 \rightleftharpoons A_3$, and $2A_1 \rightleftharpoons A_2+A_3$. The equilibrium $c^{\rm eq}$ is assumed in the center of the triangle ($c^{\rm eq}_1=c^{\rm eq}_2=c^{\rm eq}_3$). The partial equilibria of the first two reactions form the straight lines in the triangle, the medians, while the partial equilibria of the non-linear reaction form a parable $c_2c_3/c_1^2=1$. These three lines intersect in the equilibrium due to detailed balance. 

Explicit expressions for the partial equilibria $c^*_{\gamma}(c)$ are:
\begin{itemize}
\item For the reaction $A_1 \rightleftharpoons A_2$ 
\begin{equation}\label{c1gamma}
c^*_{1,\gamma_1}(c)=c^*_{2,\gamma_1}(c)=\frac{1}{2}(c_1+c_2), \; c^*_{3, \gamma_1}(c)=c_3;
\end{equation}
\item For the reaction $A_2 \rightleftharpoons A_3$
\begin{equation}\label{c2gamma}
c^*_{1, \gamma_2}(c)=c_1,\;c^*_{2, \gamma_2}(c)=c^*_{3, \gamma_2}(c)=\frac{1}{2}(c_2+c_3);
\end{equation}
\item For the reaction $2A_1 \rightleftharpoons A_2+A_3$
\begin{equation}\label{c3gamma}
\begin{split}
&c_{1, \gamma_3}^*(c)= \frac{1}{3}\left(-b_1+\sqrt{4b_1^2-3b_2^2}\right); \\
&c_{2, \gamma_3}^*(c)=\frac{1}{6}\left(4 b_1+3b_2-\sqrt{4b_1^2-3b_2^2}\right);\\
&c_{3, \gamma_3}^*(c)=\frac{1}{6}\left(4 b_1-3b_2-\sqrt{4b_1^2-3b_2^2}\right),
\end{split}
\end{equation}
where $b_1=c_1+c_2+c_3$ and $b_2=c_2-c_3$ are the independent `partial balances' for this reaction.
\end{itemize}

Fig. \ref{PatialEquilib}a  shows the partial equilibria for an arbitrarily selected point $c$. In Fig. \ref{PatialEquilib}b, one level set of $H_{\Gamma}^*$ is presented, where $\Gamma=\{\gamma_1, \gamma_2, \gamma_3\}$:
$$H_{\Gamma}^*(c)=\max\{H(c^*_{\gamma_1}(c)),H(c^*_{\gamma_2}(c)),H(c^*_{\gamma_3}(c))\},$$
$c^*_{\gamma_i}$ $i=1,2,3$ are given by (\ref{c1gamma}), (\ref{c2gamma}), (\ref{c3gamma}) and  $H$-function is the classical one (\ref{LyapFreeEN}).

 The sublevel set of $ H_{\Gamma}^*(c)$ is the  intersection of strips with sides parallel to the stoichiometric vectors $\gamma_i$. These strips are sublevel sets for the partial equilibrium entropies $H_{\gamma_i}^*(c)$. In higher dimensions, the level sets of partial equilibrium $H$-function $H_{\gamma}^*(c)$ are cylindrical hypersurfaces with the  generatrix parallel to $\gamma$. The base (or the directrix) of this cylindric surface is the level set of $H$ on the surface of partial equilibrium.  In  Fig. \ref{PatialEquilib}b, the `surfaces' of partial equilibria are lines, the level sets of $H$ on these lines are couples of points. These points are highlighted. Note that the sublevel areas for the new function $H^*_{\Gamma}$ in Fig. \ref{PatialEquilib} are convex polygons, whereas for the classical $H$-function they have smooth border.

\begin{figure}[t]
\centering{
\includegraphics[width=0.45\textwidth]{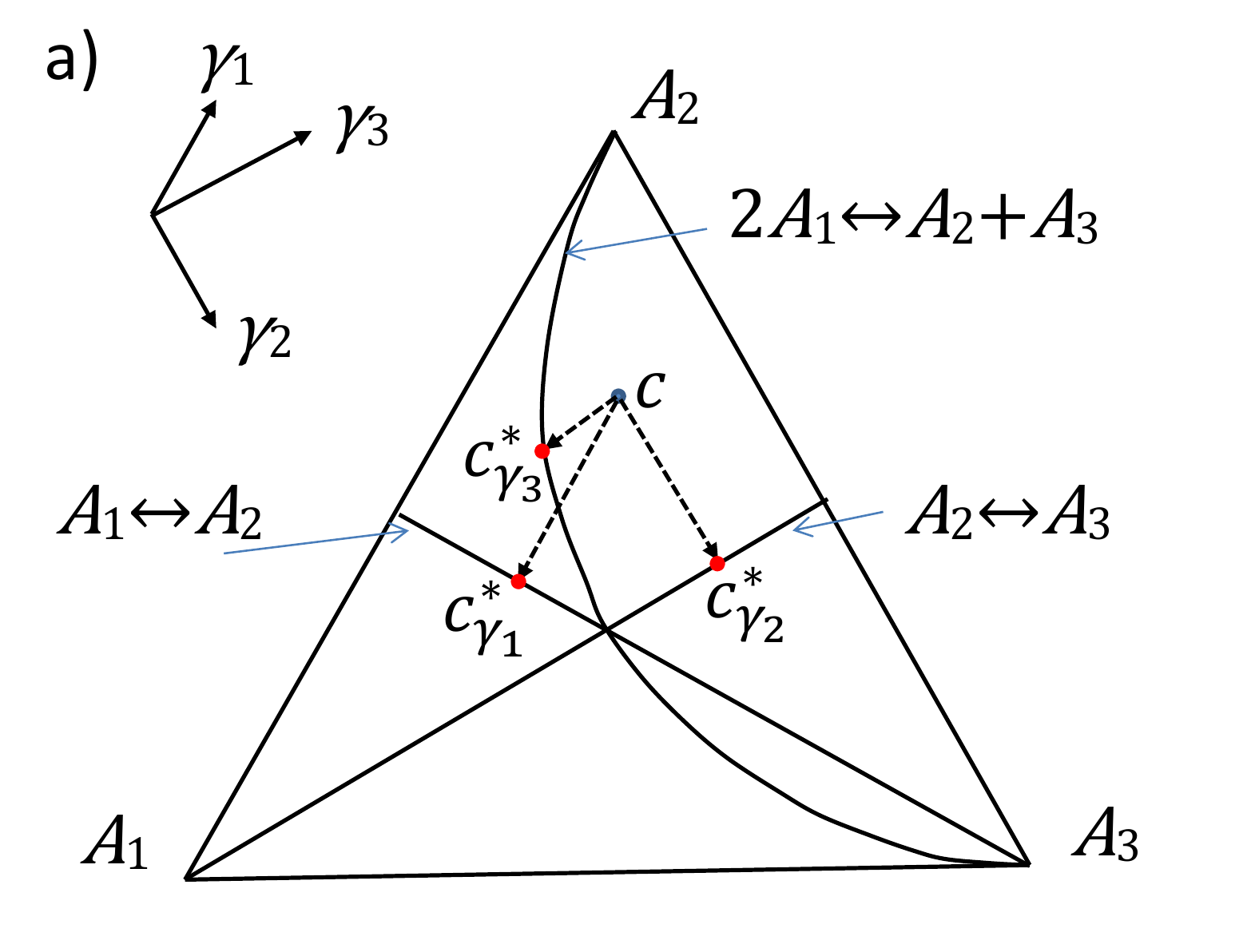} \;\;
\includegraphics[width=0.45\textwidth]{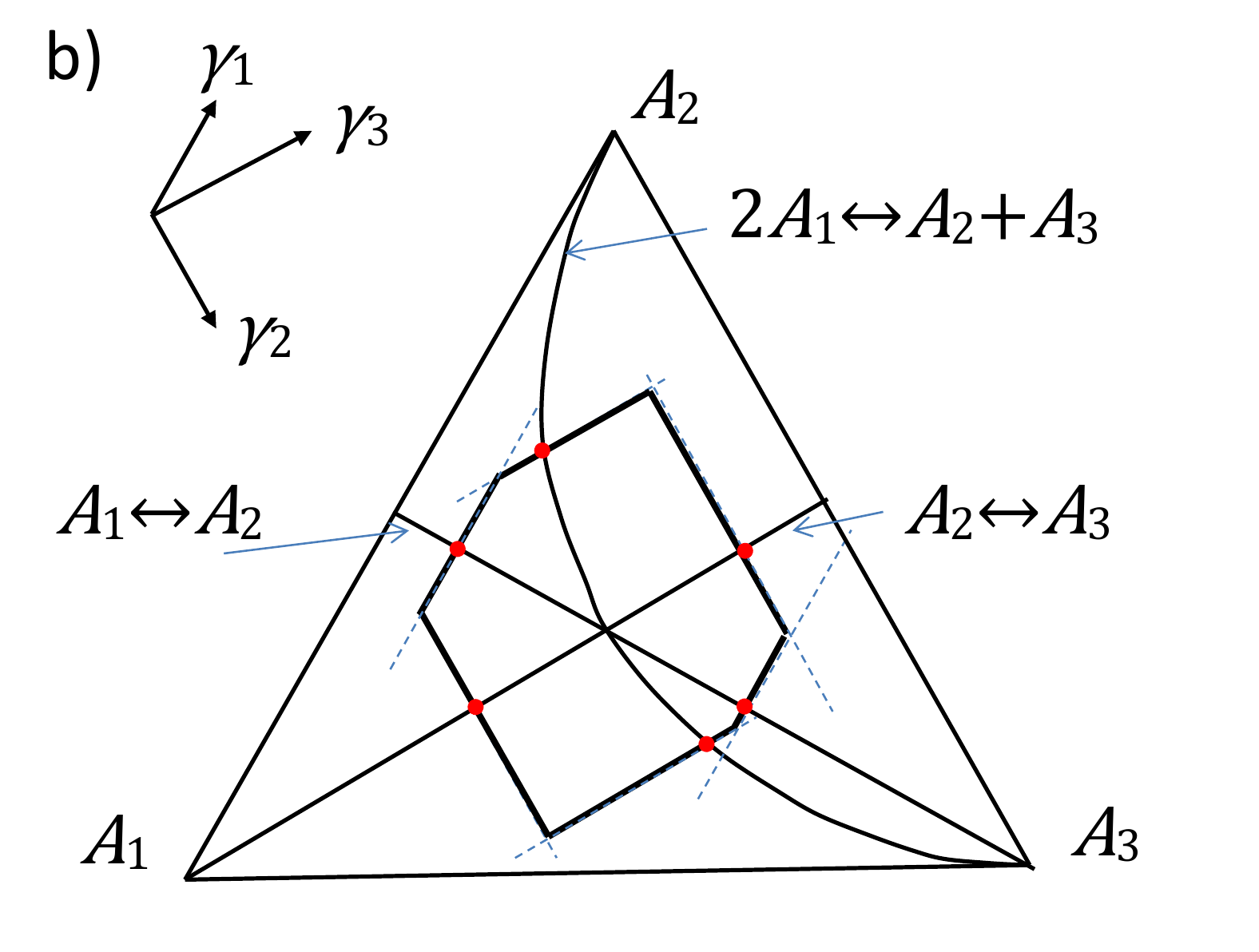}
\caption{The stoichiometric vectors $\gamma_1, \gamma_2, \gamma_3$ and the  partial equilibria for the reaction mechanism $A_1 \rightleftharpoons A_2$, $A_2 \rightleftharpoons A_3$, $2A_1 \rightleftharpoons A_2+A_3$. The concentration triangle $c_1+c_2+c_3=b$
is split by the lines of partial equilibria into six compartments. In each compartment, the dominated direction of each reaction (towards the partial equilibrium) is defined unambiguously. a) Partial equilibria (highlighted points) for an arbitrary positive concentration vector $c$.
b) A level set of $H_{\Gamma}^*$ for $\Gamma=\{\gamma_1, \gamma_2, \gamma_3\}$.  \label{PatialEquilib}}}
\end{figure}

Assume that the principle of detailed balance holds and the vector of equilibrium concentrations $c_i^{\rm eq}$ in the definition of  $H$-function (\ref{LyapFreeEN}) is the point of detailed balance.

\begin{theorem}\label{HTheoremNew}
Let all the stoichiometric vectors of the reaction mechanism (\ref{ReactMech}) belong to $\Gamma$ then $H_{\Gamma}^*(c)$ (\ref{partEQEnt1}) is monotonically non-increasing function along the solutions of the kinetic equations (\ref{reparKINUR})
 for all positive $c$ and all non-negative values of equilibrium fluxes $w^{\rm eq}$.
\end{theorem}
\begin{proof}
To use Theorem \ref{GeneralHtheorem} we have to prove two statements:
\begin{enumerate}
\item The function $H_{\Gamma}^*(c)$ is convex in  $\mathbb{R}_{>0}^n$;
\item For each $\gamma \in \Gamma$ and a positive vector $c$ the minimizer of $H$ on the interval $(c+\gamma \mathbb{R})\cap \mathbb{R}_{>0}^n$ is, at the same time, a minimizer of $H_{\Gamma}^*$ on this interval.
\end{enumerate}

We prove convexity of the function $H_{\Gamma}^*(c)$  in  $\mathbb{R}_{>0}^n$ in two steps.
\begin{itemize}
\item Convexity of $H_{\gamma}^*(c)$ for one-element sets $\Gamma= \{\gamma\}$.
\item Convexity of $H_{\Gamma}^*(c)=\max_{\gamma \in \Gamma } H_{\gamma}^*(c)$.
\end{itemize}

Let us prove convexity of $H_{\Gamma}^*(c)$ for one-element sets $\Gamma$, that is, we will prove convexity of the partial equilibrium $H$-function $H_{\gamma}^*(c)$  (\ref{partEQEnt1}). By definition, convexity of  $H_{\gamma}^*(c)$ in  $\mathbb{R}_{>0}^n$ means that for each two positive concentration vectors $c^1$ and $c^2$ and a number  $\lambda \in [0,1]$ the inequality holds:
 $$H_{\gamma}^*(\lambda c^1+ (1-\lambda) c^2) \leq \lambda H_{\gamma}^*(c^1)+(1-\lambda)H_{\gamma}^*(c^2).$$
First, notice that due to the convexity of $H$
 $$H(\lambda c_{\gamma}^*(c^1)+(1-\lambda) c_{\gamma}^*(c^2)) \leq
    \lambda H( c_{\gamma}^*(c^1))+(1-\lambda)H(c_{\gamma}^*(c^2)).$$
        Secondly, $H(c_{\gamma}^*(c^{1,2}))=H_{\gamma}^*(c^{1,2})$ by the definition of the partial equilibrium $H$-function (\ref{partEQEnt1}). Therefore, the previous inequality can be rewritten as
$$H(\lambda c_{\gamma}^*(N^1)+(1-\lambda) c_{\gamma}^*(N^2) \leq \lambda
    H_{\gamma}^*(c^1)+(1-\lambda)H_{\gamma}^*(c^2).$$
Finally, $$c_{\gamma}^*(c^{1,2}) \in (c^{1,2}+\gamma \mathbb{R})\cap \mathbb{R}_{>0}^n,$$
    hence,
$$\lambda c_{\gamma}^* (c^1)+(1-\lambda) c_{\gamma}^* (c^2)\in \lambda c^1+(1-\lambda) c^2 + \gamma \mathbb{R}$$ and
$$H(c_{\gamma}^*(c^1)+(1-\lambda) c_{\gamma}^* (c^2))\geq H_{\gamma}^*(\lambda     c^1+ (1-\lambda)c^2)$$
 because the last value is the minimum of $H$ on the interval
 $$(\lambda c^1+(1-\lambda) c^2 +\gamma \mathbb{R})\cap \mathbb{R}_{>0}^n.$$ Convexity of  the partial equilibrium $H$-function $H_{\gamma}^*(c)$  is proven.

Convexity of $H_{\Gamma}^*(c)$ follows from the convexity of the partial equilibrium $H$-function $H_{\gamma}^*(c)$, from the definition of $H_{\Gamma}^*(c)$ as the maximum of $H_{\gamma}^*(c)$ ($\gamma \in \Gamma$), and from the following fact from convex analysis: Maximum of a set of convex functions is again convex. The shortest proof is based on the definition of a convex function as a function with convex epigraph \cite{Rockafellar} and follows  from the observation that the epigraph of the maximum of a family of   functions is the intersection of their epigraphs.

Let us analyse the minimizers of $H_{\Gamma}^*$ on the interval  $I=(c+\gamma \mathbb{R})\cap \mathbb{R}_{>0}^n$ for a positive concentration vector $c$. Select $\gamma \in \Gamma$.  The minimizer of $H$ on the interval $I$ is
the partial equilibrium $c_{\gamma}^*(c)$, by the definition (\ref{partEQ1}). Function $H_{\gamma}^*(c)$ is constant on the interval $I=(c+\gamma \mathbb{R})\cap \mathbb{R}_{>0}^n$ and  $H_{\gamma}^*(c)=H(c_{\gamma}^*(c))$ on $I$. Therefore, $c_{\gamma}^*(c)$ is a minimizer of $H_{\gamma}^*(c)$ on $I$ (trivially, as all other points of $I$ do).

Notice, that for all $\gamma' \in \Gamma$
\begin{equation}\label{localMaxQE}
H_{\gamma}^*(c_{\gamma}^*(c))\geq H_{\gamma'}^*(c_{\gamma}^*(c))
\end{equation}
and the equality is strong if  $c_{\gamma}^*(c)\neq c_{\gamma'}^*(c_{\gamma}^*(c))$. Indeed, $c_{\gamma'}^*(c_{\gamma}^*(c))$ is the unique  minimizer of $H$ on the interval $(c_{\gamma}^*(c)+\gamma' \mathbb{R})\cap \mathbb{R}_{>0}^n$. If this minimizer does not coincide with $c_{\gamma}^*(c)$ then
$$H_{\gamma'}^*(c_{\gamma}^*(c))=H(c_{\gamma'}^*(c_{\gamma}^*(c)))<H(c_{\gamma}^*(c))=H_{\gamma}^*(c_{\gamma}^*(c)).$$
According to inequality (\ref{localMaxQE}),
\begin{equation}\label{MaxQEineq}
\max_{\gamma' \in \Gamma}H_{\gamma'}^*(c_{\gamma}^*(c))=H_{\gamma}^*(c_{\gamma}^*(c)).
\end{equation}

For any family of convex functions $\Phi$ on the interval $I$ the following statement holds. Let $f\in \Phi$.
If $y$ is a minimizer of $f(x)$ on $I$, $F(x)=\max\{\phi(x) | \phi \in \Phi\}$, and $F(y)=f(y)$ then $y$ is a minimizer of $F(x)$ on $I$. Indeed, for any $z\in I$ $f(z)\geq f(y)$ because $y$ is a minimizer of $f$. At the same time, $F(z)\geq f(z)$ by the definition of $F$. Hence, for any $z\in I$ $F(z)\geq F(y)$ and $y$ is a minimizer of $F(x)$ on $I$.

Let us take $f=H_{\gamma}^*(c)$ and $\Phi=\{H_{\gamma'}^*(c) | \gamma'\in\Gamma\}$. Then $F=H_{\Gamma}^*(c)$. Select $y=c_{\gamma}^*(c)$. Notice that $y$ is a minimizer of $f$ and $F(y)=f(y)$ according to (\ref{MaxQEineq}). Therefore, $c_{\gamma}^*(c)$ is a minimizer of $H_{\Gamma}^*(c)$ on $I$. By combining this result with the proven convexity of $H_{\Gamma}^*(c)$ and applying Theorem \ref{GeneralHtheorem}, we prove that $H_{\Gamma}^*(c)$ is a monotonically non-increasing function on solutions of kinetic equations.
\end{proof}

\begin{corollary}Assume that the complex balance condition holds and the vector of equilibrium concentrations $c_i^{\rm eq}$ in the definition of  $H$-function (\ref{LyapFreeEN}) is the point of complex balance.

Let all the stoichiometric vectors of the reaction mechanism (\ref{ReactMech}) belong to $\Gamma$  then the time derivative $H_{\Gamma}^*(c)$ (\ref{GreatpartEQEnt}) 
 by virtue  (\ref{reparKINUR}) is non-positive  for all positive $c$ and all non-negative values of equilibrium fluxes $w^{\rm eq}$: $\D H_{\Gamma}^*(c)/\D t \leq 0$.
\end{corollary}
\begin{proof} This corollary follows immediately from Theorem \ref{HTheoremNew} and from 
\cite[Theorem 2  (Local equivalence of detailed and complex balance)]{Gorban2014} because $\mathbb{Q}_{\rm DB}(c)=\mathbb{Q}_{\rm CB}(c)$ for all positive $c$.
\end{proof}

When $c(t)$ approaches an equilibrium, $c^{\rm eq}$, all the partial equilibria also converge to $c^{\rm eq}$. Locally, near the equilibrium $c^{\rm eq}$, this picture can be simplified. Let us use quadratic approximation to $H(c)$. The entropic inner product is defined as $\langle x|y\rangle=(x,(D^2H)_{c^{\rm eq}}y)$, where  $(D^2H)_{c^{\rm eq}}$ is the second differential of $H$ at the equilibrium (Hessian matrix) and $(\ , \ )$ is the standard inner product.

\begin{remark}[Universal Lyapunov functions for linearized kinetic equations]
 In the entropic inner product, the quadratic approximation to the classical $H$-function near equilibrium $c^{\rm eq}$ is $$H(c)-H(c^{\rm eq})\approx \frac{1}{2}\langle c-c^{\rm eq}|c-c^{\rm eq}\rangle.$$ For this approximation, elementary linear algebra gives
\begin{equation}\label{LinApprox}
\begin{split}
&c^*_{\gamma}(c)-c^{\rm eq} =  c-c^{\rm eq} - \frac{\langle c-c^{\rm eq} |\gamma\rangle \gamma}{\langle \gamma | \gamma \rangle};\\
&H_{\gamma}^*(c)-H(c^{\rm eq})= \frac{1}{2} \langle c^*_{\gamma}(c)-c^{\rm eq}| c^*_{\gamma}(c)-c^{\rm eq}\rangle = H(c)-H(c^{\rm eq})-\frac{1}{2}\frac{\langle c-c^{\rm eq} |\gamma\rangle^2}{\langle \gamma | \gamma \rangle};\\
&H_{\Gamma}^*(c)-H(c^{\rm eq})= \frac{1}{2}\langle c-c^{\rm eq}|c-c^{\rm eq}\rangle -\frac{1}{2} \min_{\gamma \in \Gamma }\left\{\frac{\langle c-c^{\rm eq} |\gamma\rangle^2}{\langle \gamma | \gamma \rangle}\right\}.
\end{split}
\end{equation}
Formula (\ref{LinApprox}) gives universal Lyapunov functions for linearized kinetic equations with the given reaction mechanism and detailed balance conditions.
\end{remark}

Thus, the new family of universal Lyapunov functions for chemical kinetic equations with detailed or complex balance is constructed.

\section{Outline of possible generalizations and applications \label{General:Sec}}

In this section, we briefly outline the possible generalizations and applications of the main results. This is basically a `to do' section and detailed analysis of all  generalizations and applications is beyond the scope of the paper.

\paragraph{Generalized mass action law} Analysis of the proven statements demonstrate that two properties of $H$ are used: strict convexity  in directions $\gamma$ in positive orthant $\mathbb{R}_{>0}$ and logarithmic singularity of the $\partial H/\partial c_i$ when $c_i\to 0$. One more requirement is the generalized mass action laws for the reaction rates:
$$w_r(c)=\phi_r \exp(\nabla H,\alpha_r),$$   
 where $ \exp(\nabla H,\alpha_r)$ is the `Boltzmann factor', and $\phi_r \geq 0$ is a non-negative quantity, `kinetic factor' \cite{GorbanShahzad2011}. 
Detailed balance and complex balance conditions for the generalized mass action law have the same form as for perfect systems. For detailed balance such kinetics was introduced by Feinberg \cite{Feinberg1972_a} (Marselin--De Donder kinetics, see also \cite{Angeli2009}). Grmela studied  properties of geometry of nonlinear non-equilibrium thermodynamics for the  generalized mass action law \cite{Grmela2012}. Detailed balance in the limit of non-equilibrium reactions was analyzed in {\cite{GorbanMirYab2013}. General analysis of non-classical entropies and their relations to the second law of thermodynamics was presented in work \cite{Gorban2014}. Complex balance conditions for macroscopic kinetics was proven for systems with Markov microscopic description in \cite{GorbanShahzad2011}. The basic constructions of partial equilibria and functions $H^*_{\Gamma}$ can be extended to the generalized mass action law.

\paragraph{Non-isochoric and non-isothermal conditions}

The Lyapunov function (\ref{LyapFreeEN}) and kinetic equations (\ref{kinurChem}) with mass action low (\ref{GenMAL}) are valid for isochoric isothermal conditions (constant volume and temperature). For other conditions they should be modified. Such a modification is rather simple and does not add any substantial change but the equation and Lyapunov functions have different analytic form (see, for example, \cite{Yab,Hangos2010}). The analysis of Lyapunov functions for these conditions can be provided in the formalism of the generalized mass action law. A general form for the description of non-isothermal reactions in closed chemical systems in terms of the Marcelin--De Donder kinetics and explicit form of the $H$-function for the systems with detailed balance under various conditions is presented in \cite{Bykov1982}.

\paragraph{Spatially distributed systems}

Transport in Boltzmann's equation is conservative (the free flight) and dissipative terms are local (the collision integral). Generalization of our approach to the collision integral has to be done. In general, models of complex transport processes can include both dissipative and conservative terms. There are many attempts to create thermodynamic theory of such processes. In the GENERIC approach, conservative and dissipative components are explicitly separated with some commutativity conditions between them \cite{GENERIC}. Generalization of proposed construction onto dissipative components of transport processes seems to be a challenging  task. If the dissipative part is described in the language of quasi-chemical formalism with a finite number of steps \cite{GorbanSargsyan2011} then this generalization is more straightforward.

\paragraph{Possible applications}

What can add the new Lyapunov functions  to research tools? Of course, more Lyapunov functions are better than less. The non-classical Lyapunov functions for linear systems are widely used for various estimates and information analysis in the situations, where linear Markov chains could serve as adequate models of information transformation. Just for example of various applications, we can refer to works  \cite{Taneja2004,Rahmani2012,Nowozin2016}. Universal Lyapunov functions are instruments for evaluation of possible dynamics when the reaction rate constants are unknown or   highly uncertain. Without any knowledge of the reaction mechanism we use the thermodynamic potentials for evaluation of the attainable sets of chemical reactions (the theory and algorithms are presented in \cite{Gorban1984,GorbanSIADS2013,FeinberHildenbrandt1997}, some industrial applications are discussed in \cite{GorbanKagan2006,Hildebrandt1990,Metzger2009}). Convexity allows us to transform the $n$-dimensional problems about attainability and attainable sets into an analysis of one-dimensional continua and discrete objects, thermodynamic trees \cite{GorbanSIADS2013}. When the reaction mechanism is known, we can use this information about mechanism for sharper estimations. A new class of estimates is needed and the new Lyapunov functions give a collection of instruments for such estimates. The use of many Lyapunov functions in the analysis of attainable sets makes these estimates more narrow and close to reality.

\section{Conclusion}

\subsection{New Lyapunov functions: the main result}

A rich family of universal Lyapunov functions for any linear or non-linear reaction network with detailed or complex balance is presented. Consider a mass action law system with a given reaction mechanism, the set of stoichiometric vectors $\Upsilon$ and detailed or complex balance. For construction of the new Lyapunov functions for this system the following operations are needed. 
 
\begin{enumerate}
\item Normalize vectors from $\Upsilon$. Let $\Upsilon_0$ be the set of normalized stoichiometric vectors.
\item Select a finite set of vectors $\Gamma$ such that $\Upsilon_0 \subset \Gamma$ and each vector from $\Gamma$ has both positive and negative components.
\item For each  $\gamma \in \Gamma$ calculate $$H^*_{\gamma}(c)=\min_{c+\gamma x \in\mathbb{R}_{>0}^n}H(c+ \gamma  x).$$
\item Find $$H^*_{\Gamma}(c)=\max_{\gamma \in \Gamma}H^*_{\gamma}(c).$$
\end{enumerate}
Thus, we have to perform several smooth one-dimensional convex minimizations and then select the maximum of these minima.

Theoretically, $H^*_{\Gamma}$  is a Lyapunov function for any {\em compact} $\Gamma$. Such a function can be approximated (from below) by the functions calculated for finite $\varepsilon$-networks on $\Gamma$. 

 Functions $H^*_{\Gamma}$  are not presented as elementary functions even if $H$ has explicit analytical form. If the reaction mechanism consists of mono- and bimolecular reactions  then the partial equilibria can be represented by radicals and $H^*_{\gamma}(c)$ can be written as an elementary function but in general case the less implicit form with min and max operations is necessary. If we substitute $H$ by its quadratic approximation near equilibrium then functions $H^*_{\Gamma}$  can be represented more explicitly (\ref{LinApprox}) and the difference between $H$ and $H^*_{\Gamma}$ becomes obvious.

The mystery about the  fundamental difference between the rich family of Lyapunov functions for linear networks and a very limited collection of Lyapunov functions for non-linear networks in thermodynamic conditions is resolved: there is no such crucial difference anymore.

\subsection{An open problem}

There remains an obvious difference between explicit analytic expression of $f$-divergences (\ref{F-div}) and not so obvious construction of Lyapunov functions for general networks using partial equilibria of non-linear reactions. For linear reactions, the partial equilibria have very  simple analytic expression, for bimolecular reactions they are given explicitly using quadratic formula, but for trimolecular reactions the analytic formulas become too bulky.

Alt least one important question is still open. The new Lyapunov functions $H_{\Gamma}^*$ are, at the same time, universal Lyapunov functions for linear kinetics, if the stoichiometric vectors of the linear reaction mechanism $A_i \rightleftharpoons A_j$ are included in $\Gamma$. Due to the results of \cite{ENTR3,Amari2009}, such a function should be, essentially,  a $f$-divergence (\ref{F-div}) $H_f(c)$, or, more precisely, it should be a monotonic function of $H_f(c) +\lambda_i c_i$ for some constant $\lambda$. Nevertheless, now we know nothing about these $f$-divergences except their existence. Constructive transformation of   $H_{\Gamma}^*$ into $f$-divergence is desirable because an explicit form (\ref{F-div}) brings some benefits for analysis.

\section*{Acknowledgments}
The work was supported by the University of Leicester and the Ministry of  Science and Higher Education of the Russian Federation (Project No. 14.Y26.31.0022).

\end{document}